\documentclass[11pt]{article}

\advance \topmargin by -\headheight
\advance \topmargin by -\headsep

\evensidemargin \oddsidemargin
\usepackage{amsmath, amssymb, amsthm}
\usepackage{algorithmic}
\usepackage{algorithm}

\newcommand{\E}{\mathrm{E}}

\newcommand{\bR}{\mathbb{R}}

\newcommand{\ie}{{\it i.e.}}

\newtheorem{theorem}{Theorem}[section]
\newtheorem{corollary}[theorem]{Corollary}
\newtheorem{observation}[theorem]{Observation}

\newtheorem{proposition}[theorem]{Proposition}

\newcommand\ignore[1]{}
\newtheorem{fact}{Fact}

\newtheorem{remark}[theorem]{Remark}

\newtheorem{definition}[theorem]{Definition}

\title{Selling Privacy at Auction}
\author{Arpita Ghosh \\ Yahoo Research \\ 2821 Mission College Blvd. \\ Santa Clara, CA \thanks{arpita@yahoo-inc.com} \and Aaron Roth \\ CIS Department \\ University of Pennsylvania \\ 3330 Walnut Street \\ Philadelphia, PA \thanks{aaroth@cis.upenn.edu}}
\begin{document}
\maketitle

\begin{abstract}
We initiate the study of markets for private data,  using differential privacy.  We model a data analyst who wishes to buy sensitive information to estimate a population statistic.  The analyst wishes to obtain an \emph{accurate} estimate \emph{cheaply}, while the owners of the private data experience cost for their loss of privacy.

Our main result is that this problem can naturally be viewed and optimally solved as a variant of a multi-unit procurement auction. We derive auctions which are optimal up to small constant factors for two natural settings:
 \begin{enumerate}
 \item A data analyst with a fixed accuracy target, wishing to minimize his payments.
 \item A data analyst with a hard budget constraint, wishing to maximize his accuracy.
 \end{enumerate}
 In both cases, our comparison class is the set of envy-free pricings.

We then define a more stringent privacy model,  and show that  no individually rational mechanism in this model can achieve non-trivial accuracy. We propose several directions for future research to remedy this situation.
\end{abstract}
\thispagestyle{empty}
\setcounter{page}{0}
\pagebreak

\section{Introduction}
Organizations such as the Census Bureau and hospitals have long maintained databases of personal information. However, with the advent of the Internet, many corporations are now able to aggregate enormous quantities of sensitive information, and use, buy, and sell it for financial gain. Up until recently, the purchase and sale of private information was the exclusive domain of aggregators -- it was obtained for free from the actual owners of the data, for whom it was sensitive. However, recently, companies such as ``mint.com'' and ``Bynamite'' have started acting as brokers for private information at the consumer end, paying users for access to their sensitive information \cite{newsarticle1,newsarticle2}.
Many others, such as Yahoo, Microsoft, Google, and Facebook are also implicitly engaging in the purchase of private information in exchange for non-monetary compensation. In short, ``privacy'' has become a commodity that has already begun to be bought and sold, in a variety of ad-hoc ways.

Despite the commoditization of privacy in practice, markets for privacy lack a theoretical foundation. In this paper, we initiate the rigorous study of markets \emph{for} private data. Our goal is not to provide a complete solution for the myriad problems involved in the sale of private data, but rather to introduce a crisp model with which to investigate some of the many questions unique to the sale of private data. The study of privacy as a commodity is of immediate relevance, and also a source of many interesting theoretical problems: we hope that this paper elicits more new questions than it answers.

First, let us briefly consider some of the issues that make \emph{privacy} distinct from other commodities that we often deal with, and why this may complicate its sale:
\begin{enumerate}
\item First and foremost, in order sell \emph{privacy}, it is important to be able to define and quantify what we mean by privacy. In this regard, the commoditization of privacy has dovetailed nicely with the development of the theoretical underpinnings of privacy: recent work on \emph{differential privacy} \cite{DMNS06} (Definition \ref{diff-privacy}) provides a compelling definition and a precise way in which to quantify its sale. Importantly, as we will discuss, the guarantee of differential privacy has a natural utility-theoretic interpretation that makes it a natural quantity to buy and sell.
\item Private data is a good that exhibits intrinsic complementarities: a data analyst will typically not be interested in the private data of any particular individual, but rather in a representative sample from a large population. Nevertheless, he must purchase the data from particular individuals! Clearly, if there may be unknown correlations between individuals values for privacy and their private data, then the typical strategy of ``buying from the cheapest sellers'' is doomed to fail in this regard. How should an auction be structured by an analyst who wishes to calculate some value which is representative of an entire population?
\item An individual's cost for privacy may itself be private information. Suppose that Alice visits an oncologist, and subsequently is observed to significantly increase her value for privacy: this is of course disclosive! Is it possible to run an auction for private data that compensates individuals for the privacy loss they incur, simply due to the effect that their bids have on the behavior of the mechanism?
\end{enumerate}

\subsection{Differential Privacy as a Commodity}
Differential privacy, formally defined in Section \ref{sec:prelims}, was introduced by Dwork et al. \cite{DMNS06} as a technical definition for database privacy. Informally, an algorithm is $\epsilon$-differentially private if changing the data of a single individual does not change the probability of any outcome of the mechanism by more than an $\exp(\epsilon) \approx (1+\epsilon)$ multiplicative factor. Differential privacy also has a natural utility-theoretic interpretation that makes it a compelling measure with which to quantify privacy when buying or selling it\footnote{This utility theoretic interpretation has been used in another context: the work of McSherry and Talwar, and Nissim, Smorodinsky and Tennenholtz \cite{MT07, NST10} using differential privacy as a tool for traditional mechanism design.}.

 An important property of an $\epsilon$-differentially private algorithm $A$ is that its composition with any other database-independent function $f$ has the property that $f(A)$ remains $\epsilon$-differentially private. This allows us to reason about events that might seem quite far removed from the actual output of the algorithm. Quite literally, a guarantee of $\epsilon$-differential privacy is a guarantee that the probability of receiving phone calls during dinner, or of being denied health insurance will not increase by more than an $\exp(\epsilon)$ factor. This allows us to interpret differential privacy as a strong utility theoretic guarantee that holds simultaneously for arbitrary, unknown utility functions: for any individual, with any utility function $u$ over (arbitrary) future events, an $\epsilon$-differentially private computation will decrease his future expected utility by at most an $\exp(-\epsilon) \approx (1-\epsilon)$ multiplicative factor, or equivalently, by an $\epsilon\E[u(x)]$ additive factor, where the expectation is taken over all future events that the individual has preferences over. Therefore, there is a natural way for an individual to assign a cost to the use of his data in an $\epsilon$-differentially private manner: it should be worth to him an $\epsilon$-fraction of his expected future utility. We expand on this in section \ref{app:justification}.


\subsection{Results}
Our main contribution is to show that any differentially private mechanism that guarantees a certain accuracy {\em must} purchase a certain minimum amount of privacy from a certain minimum number of agents (both of which depend on the desired accuracy), which reduces the problem of privately providing an accurate answer to a relatively simple form of procurement problem.
Specifically, we study the following stylized model. There are $n$ individuals $[n]$, each of whom possesses a private bit $b_i$, which is already known by the administrator of the private database (for example, a hospital). Each individual also has a certain \emph{cost} function $c_i :\bR_+\rightarrow \bR_+$, which determines what her cost $c_i(\epsilon)$ is for her private bit $b_i$ to be used in an $\epsilon$-differentially private manner. Any feasible mechanism must pay each individual enough to compensate him for the use of his private data. Moreover, individuals may mis-report their cost functions in an attempt to maximize their payment, and so we are interested in mechanisms which properly incentivize individuals to report their true cost for privacy. On the other side of the market, the \emph{data analyst} wishes to estimate the quantity $s = \sum_{i=1}^n b_i$, and must compensate each individual through the mechanism's payments for this estimate. The data analyst may either have a fixed accuracy objective and wish to minimize his payments subject to obtaining the desired accuracy, or alternately have a fixed budget and wish to maximize the accuracy of his estimate within this budget.

We first consider the simpler model, in which individuals must be compensated for loss of privacy to their bits $b_i$, but not for any privacy-leakage due to implicit correlations between $b_i$ and their cost function $c_i$ (\ie, if the mechanism does not use an individual's bit $b_i$ at all in computing an estimate for the data analyst, the mechanism does not have to compensate individual $i$, even if changing her cost function would result in a different outcome for the mechanism). In trying to design an auction that guarantees the data analyst an accurate estimate of $s$, one might consider any number of complicated mechanisms that (for example) randomly sample individuals, and then attempt to buy from entire random samples -- there are many variations therein, and indeed, this was the direction from which we first explored the problem. Our main result is that it is not necessary to consider such mechanisms. We show that we may abstract away the structure of the mechanism, and without loss of generality consider multi-unit procurement auctions. This has some immediate consequences: if we are interested in the setting for which the data analyst has a fixed accuracy goal, subject to which he wishes to minimize his payment, then we show that the standard VCG mechanism is optimal among the set of envy-free mechanisms. If we are instead interested in the setting for which the data analyst has a fixed budget subject to which he wishes to maximize his accuracy, then we are in a more unusual procurement-auction setting: the buyer wishes to maximize the number of sellers he can buy from, and the cost to the sellers is a function of who else sells their data! In this setting, we give a truthful mechanism that is instance-by-instance optimal among the set of all fixed-price (envy free) mechanisms. We remark that our choice of fixed-price mechanisms as a benchmark has become standard in prior-free mechanism design (see, e.g. \cite{HK07, HR08}), but stands on firmer ground in auction settings for which Bayesian-optimal mechanisms are known also to charge fixed prices. We operate in a setting in which Bayesian-optimal mechanisms are not known, and so justifying (or improving) this choice of benchmark in our setting is an interesting open problem. (We note that \cite{EG11} derives the Bayesian optimal auction for a budget-constraint buyer who wants to purchase a set of items with maximum value subject to her budget constraint, but the model in \cite{EG11} does not capture our procurement auction problem with externalities which arise because the amount of privacy that needs to be bought from an individual (and therefore the cost to that seller) depends on the total number of individuals from whom privacy is being purchased.)

We then show a generic impossibility result: it is not, in general, possible for any direct revelation mechanism to compensate individuals for their privacy loss due to unknown correlations between their private bits $b_i$ and their cost functions $c_i$. If their costs are known to lie in some fixed range initially, it is  possible to offer them some non-trivial privacy guarantee, but finding the correct model in which to study the issue of unknown correlations between data and valuation for privacy is another important direction in which to take this research agenda.

\subsection{Related Work}
\subsubsection{Differential Privacy and Mechanism Design}
McSherry and Talwar proposed that differential privacy could itself be used as a \emph{solution concept} in mechanism design \cite{MT07}. They observed that a differentially private mechanism is approximately truthful, while simultaneously having some resilience to collusion. Using differential privacy as a solution concept as opposed to dominant strategy truthfulness, they gave some improved results in a variety of auction settings. Gupta et al. also used differential privacy as a solution concept in auction design \cite{GLMRT10}.

In a beautiful follow-up paper, Nissim, Smorodinsky, and Tennenholtz \cite{NST10} made the point that differential privacy may not be a compelling solution concept when beneficial deviations are easy to find (as indeed they are in the mechanism of \cite{MT07}). Nevertheless, they demonstrated a generic methodology for using differentially private mechanisms as tools for designing exactly truthful mechanisms \emph{that do not require payments}, and demonstrate the utility of this framework by designing new mechanisms for several problems.

In this paper, we consider an orthogonal problem: we do not try to use differential privacy as a tool in traditional mechanism design, but instead try to use the tools of traditional mechanism design to \emph{sell} differential privacy as a commodity. Nevertheless, we also use the utility theoretic properties of differential privacy that allow McSherry and Talwar to prove that it implies approximate truthfulness to motivate why it is natural for individuals to have linear cost functions for differential privacy.

In very recent work, Xiao~\cite{X11} addresses another question at the intersection of differential privacy and mechanism design: suppose the output of the database sanitization mechanism can be interpreted both as a sensitive quantity that must satisfy a differential privacy requirement, as well as the outcome of a {\em game}, the utility from which motivates agents to participate in the database in the first place. With this interpretation, it is reasonable to imagine that participants may lie about the bit stored in the database itself, in order to improve their utility from this game.
Xiao shows how to construct mechanisms that are simultaneously exactly truthful and differentially private, but also shows that this conjunction of truthfulness and differential privacy may not be sufficient to elicit truthful behavior when agents value privacy, \ie, have a cost to the information leaked by the mechanism about their private bit.
Chen et al. \cite{CCKMV11} propose a new, more general way of measuring privacy in agents' utility functions than that in \cite{X11}, and construct mechanisms that are truthful when including this privacy measure in the agents' utilities for settings that include 2-candidate voting, discrete facility location,  and the Groves mechanisms for public projects.
The key differences between our work and this line of investigation arises from what is treated as private information, \ie, what agents can lie about--- the agents in our model cannot lie about their private bit, which is already known to the database, but can strategically report their costs for privacy to increase their payment from the analyst, whereas the agents in \cite{X11} can lie about their private data to improve their utility from a game whose outcome depends on this input.

\subsubsection{Auctions Which Preserve Privacy}
Recently, Feigenbaum, Jaggard, and Schapira considered (using a different notion of privacy) how the implementation of an auction can affect how many bits of information are leaked about individuals bids \cite{Feigenbaum}. Specifically, they study to what extent information must be leaked in second price auctions and in the millionaires problem. Protecting the privacy of bids is an important problem, and although it is not the main focus of this paper, we consider it in the context of differential privacy in Section \ref{sec:privatebids}. We consider somewhat orthogonal notions of privacy and implementation that make our results incomparable to those of \cite{Feigenbaum}.

\subsubsection{Privacy and Economics}
Privacy and its relation to mechanism design has also been studied from a broader economic perspective, although primarily in the context of how preferences for privacy by agents may affect mechanisms, rather than in the context of markets \emph{for} privacy. For example, Calzolari and Pavan study the optimal disclosure policy when designing contracts for buyers who are in the position of repeatedly choosing between multiple sellers \cite{CP06}, and the recent work of  Taylor, Conitzer, and Wagman \cite{TCW10} studies the relationship between the ability of consumers to keep their identity private, and the ability of a monopolist to engage in price discrimination.

An exception is the essay of Laudon~\cite{Lau96}, which proposes the idea of a market for personal
information--- a `National Information Market'--- where individuals can
choose to sell or lease their information (possibly to be used in
aggregation with other individuals' information) in exchange for a
share of the revenue generated from its use; he argues that only
individuals whose cost from the `annoyance' caused by releasing their
information is lower than the payment they receive will participate in
this market. In the same spirit, the work of Kleinberg, Papadimitrou and Raghavan~\cite{KPR01} quantifies the value of private information in some specific settings, and proposes that individuals should be compensated for the use of their information to the extent of this value. Our individually rational auctions for privacy are
conceptually similar to this, but are investigated within the formal
framework of differential privacy, and from the perspective of
auction design.

\subsubsection{Relationship to the Differential Privacy Literature}
The now large literature on differential privacy (see \cite{DworkSurvey} for an excellent overview) has almost exclusively focused on techniques for guaranteeing $\epsilon$-differential privacy for various tasks, where $\epsilon$ has been taken as a given parameter. What has been almost entirely missing is any normative guidance for how to pick $\epsilon$. There is a natural tradeoff between the privacy parameter $\epsilon$ and the accuracy of privacy-preserving estimates (which is well-understood in the case of single statistics, see \cite{GRS09, BN10}). Therefore, this paper proposes to answer the question of how $\epsilon$ should be chosen: it should be the smallest value that the data analyst is able to afford, given the individuals' valuations for privacy (or equivalently, the smallest value that the owners of the data are willing to accept in exchange for their payment).

We also highlight in this work the explicit tradeoff between compensating individuals for the use of their private information, and the accuracy of our resulting estimates. Implicit in previous works on privacy has been the idea that for fixed values of $\epsilon$, individuals should be willing to participate in private databases given only some small positive incentive. However, this incentive may be different for different individuals, and without running an auction, a data collector is engaging in selection bias: he is only collecting data from those individuals who value their privacy at a low enough level to make participation in a given database worth while. Such individuals might not be representative of the general population, and resulting estimates may therefore be inaccurate. This source of inaccuracy is hidden in previous works, but we point out that it should be a real concern, and we explicitly address it in this paper.

\section{Preliminaries}
\label{sec:prelims}
We consider a database consisting of the data of $n$ individuals $\{1,\ldots,n\}$ whom we denote by $[n]$. Each individual $i$ is associated with a private bit $b_i \in \{0,1\}$, as well as a value $v_i$ parameterizing a cost function which quantifies their cost for loss of privacy. (We may think of the private bit as representing the answer to some arbitrary yes or no question. For the sake of discussion, let us assume that the private bit represents whether the individual has some embarrassing medical condition.) The private bit $b_i$ is verifiable, and the individual is not endowed with the ability to lie about their private bit. For example, the bit may already be known to a trusted database administrator (for example a hospital), or may be directly verifiable by the auctioneer (e.g. through a blood or saliva sample). On the other hand, the individual may lie about their \emph{value} for privacy $v_i$, and must be incentivized to report this parameter truthfully. We formalize this model in the following section.

\subsection{Differential Privacy}
Formally, a data set or database $D$ of size $n$ is a collection of $n$ elements from some abstract range $X$: $D \in X^n$.  We think of each element in the database as corresponding to the data of a single individual. Two databases $D,D^{(i)} \in X^n$ are \emph{neighbors} if they differ only in the data of a single individual, \ie, if $D_j = D^{(i)}_j$ for all $j \neq i$. The quantification of privacy we employ is that of \emph{differential privacy}, due to Dwork et al. \cite{DMNS06}:

\begin{definition}\label{diff-privacy}
An algorithm $A:X^n\rightarrow R$ (for an abstract range $R$) satisfies $\epsilon_i$-differential privacy with respect to individual $i$ if for any pair of neighboring databases $D, D^{(i)} \in X^n$ differing only in their $i$'th element, and for any event $S \subset R$:
$$\frac{\Pr[A(D) \in S]}{\Pr[A(D^{(i)}) \in S]} \leq e^{\epsilon_i}$$
An algorithm $A$ is $\epsilon_i$-{\em minimally private} with respect to individual $i$ if $\epsilon_i = \inf \epsilon$ such that $A$ is $\epsilon$-differentially private with respect to individual $i$. Throughout this paper, whenever we say that an algorithm is $\epsilon_i$-differentially private, we mean that it is $\epsilon_i$-minimally differentially private.
\end{definition}
\begin{remark}
A couple of remarks are in order. First note that (unless $A$ computes a constant function) for $A$ to be differentially private it must be a randomized algorithm. Second, note that differential privacy states intuitively that no single individual can have a large effect on the output distribution of an algorithm $A$, and hence the output of $A$ contains little ``information'' about any individual. Indeed, if $B$ is a random variable taking values in $X^n$, stating that $A$ is $\epsilon$-differentially private with respect to each individual $i$ is a stronger guarantee (and in particular implies) that the mutual information between $B$ and $A(B)$ is at most $\epsilon$: $I(B; A(B)) \leq \epsilon$. In particular, note that as a privacy guarantee,
$\epsilon$-differential privacy becomes less meaningful for large values of $\epsilon$. In this paper, we will restrict our attention to values of $\epsilon < 1$. Note that in this case, $\exp(\epsilon) \approx 1+\epsilon$.
\end{remark}

The following easy fact follows immediately~\cite{DMNS06}:
\begin{fact}
\label{fact:compose}
Consider an algorithm $A:X^n\rightarrow R$ that satisfies $\epsilon_i$-differential privacy with respect to each individual $i$, and let $T \subset [n]$ denote a set of indices. Consider two databases $D, D^T \in X^n$ at Hamming distance $|T|$ that differ exactly on the indices in $T$. Then for any event $S \subseteq R$:
$$\frac{\Pr[A(D) \in S]}{\Pr[A(D^{T}) \in S]} \leq e^{\sum_{i \in T}\epsilon_i}$$
\end{fact}
\begin{proof}
Consider the sequence of databases $D^0,\ldots,D^{|T|}$ such that $D^0 = D$, $D^{|T|} = D^T$  and for each $0 \leq i < |T|$, databases $D^i$ and $D^{i+1}$ are neighbors, differing in exactly the $i$'th index of $T$. Then  for any event $S$:
$$\frac{\Pr[A(D) \in S]}{\Pr[A(D^T) \in S]} = \prod_{i=0}^{|T|-1}\frac{\Pr[A(D^i) \in S]}{\Pr[A(D^{i+1}) \in S]} \leq \prod_{i \in T}\exp(\epsilon_i) = \exp(\sum_{i\in T}\epsilon_i)$$
\end{proof}

A useful primitive for differential privacy is the Laplacian distribution, adding random noise from which produces differentially private output~\cite{DMNS06}:
\begin{definition}
Denote by Lap$(\sigma)$ the symmetric Laplacian distribution with mean $0$ and scaling $\sigma$. This distribution has probability density function:
$$f(x) = \frac{1}{2\sigma}\exp\left(-\frac{|x|}{\sigma}\right)$$
We will sometimes abuse notation and write Lap($\sigma$) to denote the realization of a random variable drawn from the Laplacian distribution with parameter $\sigma$.
\end{definition}

\subsection{Mechanism Design}
In this section, we specify the utility function of the participants in the mechanism, and in particular, how it relates to the privacy guarantees of the mechanism.

Every individual $i$ has some (unknown to the mechanism) single-parameter cost function with parameter $v_i$, written  $c(v_i, \cdot):\bR_+\rightarrow \bR_+$, where $c(v_i, \epsilon)$ represents player $i$'s cost for having his data used in an $\epsilon$-differentially private manner. Cost functions are normalized so that $c(v_i,0) = 0$ for all $v_i$, and are assumed to be continuous. We will study two models, informally, one of which will treat an individuals data only as his private bit $b_i$, and one of which will treat his data as the tuple $(v_i,b_i)$.  Each individuals cost function belongs to the same publicly known family, but the parameter $v_i$ is known only to the individual, and must be reported to the mechanism. We will require that the family of cost functions admit a total ordering independently of $\epsilon$. That is, the property that our results will require is that for any $i \neq j$, and for any $\epsilon \in \bR_+$, it should hold that $c(v_i,\epsilon) \leq c(v_j, \epsilon)$ if and only if $v_i  \leq v_j$. Natural choices of cost functions which obey this property are linear cost functions, which take the form $c(v_i,\epsilon) = v_i\epsilon$, exponential cost functions which take the form $c(v_i,\epsilon) = \exp(\epsilon v_i)$, quadratic cost functions of the form $c(v_i,\epsilon) = v_i\epsilon^2$, as well as many other natural choices.

A mechanism $M:\bR_+^n\times \{0,1\}^n\rightarrow \bR \times \bR_+^n$ takes as input a vector of cost parameters $v = (v_1,\ldots,v_n) \in \bR_+^n$ and a collection of private bit values $b \in \{0,1\}^n$, and outputs a real number (an estimate of some statistic $\hat{s}$ of $b$ of interest to the ``data analyst''), as well as a payment that will be collected from the data analyst to be distributed to the participants in the mechanism.

We consider two models of privacy:
\begin{enumerate}
\item In the \emph{insensitive value} model, the mechanism $M$ first inspects the reported cost parameters $v$ and then selects (as an arbitrary deterministic function of $v$) a randomized algorithm $A:\{0,1\}^n\rightarrow \mathbb{R}$ together with a set of payments $p_1,\ldots,p_n$. The mechanism then computes $A(b)$, which is $\epsilon_i$-differentially private with respect to each individual $i$, for some $\epsilon_i$. $M$ outputs a statistic $\hat{s} = A(b)$, and each individual $i$ experiences cost $c(v_i,\epsilon_i)$. Note that in this model, individuals incur privacy cost only as a function of the use of their private bit $b_i$, and not as a function of the use of their value for privacy $v_i$: the algorithm is free to use $v_i$ in any way. The algorithm pays each individual $p_i$, and collects $P \geq \sum_{i=1}^np_i$ from the data analyst. Note that the output of the mechanism from the data analyst's perspective is the pair $(\hat{s}, P)$.
\item In the \emph{sensitive value} model, $M$ is some randomized algorithm $M:X^n\rightarrow R$, for $X = \bR_+\times \{0,1\}$ and $R = \bR_+\times \bR_+$: i.e. its input is a set of n tuples $(v_i,b_i)$ one for each individual, and its output is a pair of reals: a statistic $\hat{s}$ and a payment collected from the data analyst $P$. $M$ itself is $\epsilon_i$-differentially private with respect to each individual $I$ for some $\epsilon_i$, and each individual experiences cost $c(v_i,\epsilon_i)$. The mechanism compensates each individual an (unobserved) amount $p_i$, with the restriction that $\sum_{i=1}^n p_i \leq P$. Note that in this model, individuals experience cost as a function of the use of both their private bit $b_i$, as well as their reported value for privacy $v_i$.
\end{enumerate}

Note that in both cases, the data analyst only learns the total payment $P$ that he must make, not necessarily the distributions $p_i$ which are made to each individual.

 For any $v'_i \in \bR_+$ we let $(v_{-i}, v'_i)$ denote the vector that results from changing entry $v_i$ in $v$ to $v_i'$.

A player $i$ who recieves payment $p_i$, and whose data is used in an $\epsilon_i$-differentially private way derives utility $u_i = p_i - c(v_i,\epsilon_i)$. Here, $\epsilon_i$ is the privacy parameter of the selected algorithm $A$ if we are in the insensitive value model, and is the privacy parameter of the mechanism $M$ if we are in the sensitive value model.

 Since any individual may opt against participating in our mechanism, we require first that our mechanisms be \emph{individually rational}:
\begin{definition}
A mechanism $M:\bR_+^n\times \{0,1\}^n\rightarrow \bR \times \bR_+^n$ is \emph{individually rational} if for all $v \in \bR_+^n$:
$$p_i(v) \geq c(v_i,\epsilon_i(v))$$
 where here the payment vectors and privacy parameters are viewed as functions of the vector of reported types $v$. If $p_i(v)$ and $\epsilon_i(v)$ are random variables, we require that our mechanisms are ex-post individually rational: that is, the above inequality must hold for all realizations of $p_i(v)$ and $\epsilon_i(v)$. In words, each player must be guaranteed non-negative utility by participating and truthfully reporting his value to the mechanism.
\end{definition}
Since individuals may misreport their costs so as to maximize their gain, we also require our mechanisms to be \emph{truthful}:
\begin{definition}
A mechanism $M:\bR_+^n\times \{0,1\}^n\rightarrow \bR \times \bR_+^n$ is \emph{dominant-strategy truthful} if for all $v \in \bR_+^n$, for all $i \in [n]$, and for all $v_i' \in \bR_+$:
$$p_i(v) - c(v_i,\epsilon_i(v)) \geq p_i(v_{-i}, v_i') - c(v_i, \epsilon_i(v_{-i}, v_i'))$$
that is, no player can ever increase his utility by misreporting his value for privacy. If $p_i(v)$ and $\epsilon_i(v)$ are random variables, the above inequality should hold in expectation over the internal randomness of the mechanism\footnote{That is, we require only truthfulness in expectation. However, all of our mechanisms will in fact be ex-post truthful, and in fact the payment schemes will be deterministic. Our lower bounds will hold even for mechanisms which are merely truthful in expectation, which only strengthens our results.}.
\end{definition}

The mechanism is run on behalf of some data analyst, who wishes to know an estimate of the statistic $s \equiv \sum_{i=1}^n b_i$. The mechanism outputs some randomized estimate of this quantity $\hat{s}$, where the randomization is to ensure differential privacy, and the analyst prefers more accurate answers. We choose to focus on statistics which can be represented as sums of boolean variables because of the central role that they play in the privacy literature (in which they are known as \emph{counting queries} or \emph{predicate queries}). In particular, the ability to accurately answer queries of this sort is sufficient to be able to implement a wide range of machine learning algorithms over the data (see \cite{BDMN05}).
\begin{definition}
A mechanism $M$ satisfies $k$-accuracy if for any $D \in \{0,1\}^n$, it outputs an estimate $\hat{s}$ such that:
$$\Pr[|\hat{s} - s| \geq k] \leq \frac{1}{3}$$
where the probability is taken over the internal coins of the mechanism.
\end{definition}
The constant $1/3$ is of course inconsequential: it can be changed to any desired constant without qualitatively affecting the results.

We may consider two dual objectives for our mechanism. Our data analyst may have a fixed goal of $k$-accuracy for some $k$ in which case we want to design mechanisms which deliver $k$-accurate estimates of $s$ so as to minimize the sum of the payments. Alternately, our data analyst may have a fixed budget $B \in \bR_+$ (say an NSF grant that can be used for data procurement). In this case, our goal is to design a mechanism which is $k$-accurate for the smallest possible value of $k$, while under the constraint that the sum of the payments never exceeds $B$.

\subsection{Valuing Differential Privacy}
\label{app:justification}
In this section, we provide a brief justification for why individuals should be able to quantify their cost for experiencing an $\epsilon$-differentially private use of their private data.
Say that $\mathcal{A}$ denotes the set of all future events for which an individual $i$ has preferences over outcomes, and $u_i:\mathcal{A}\rightarrow \bR$ is a function mapping events to $i$'s utility for that event. Suppose that $D \in X^n$ is a data-set containing individual $i$'s private data, and that $M:X^n\rightarrow R$ is a mechanism operating on $D$ promising $\epsilon_i$-differential privacy to individual $i$. Let $D'$ be a data-set that is identical to $D$ except that it does not include the data of individual $i$ (equivalently, it includes the data of individual $i$, but it is used in a $0$-differentially private manner), and let $f:R\rightarrow \Delta\mathcal{A}$ be the (arbitrary) function that determines the distribution over all future events, conditioned on the output of mechanism $M$.

A basic consequence of differential privacy is the following:
\begin{fact}
\label{fact:postprocessing}
If $M:X^n\rightarrow R$ is $\epsilon_i$-differentially private with respect to individual $i$, and $f:R\rightarrow R'$ is any arbitrary (randomized) function, then the composition $f\circ M:X^n\rightarrow R'$ is also $\epsilon_i$-differentially private with respect to individual $i$.
\end{fact}
\begin{proof}
First, assume that $f$ is a deterministic function $f:R\rightarrow R'$. Fix any event $S \subset R'$ and let $T \subset R$ be $T = \{r \in R : f(r) \in S\}$. Now for any pair of neighboring databases $D, D' \in X^n$ differing in their $i$'th coordinate, we have:
\begin{eqnarray*}
\Pr[f(M(D)) \in S] &=& \Pr[M(D) \in T] \\
&\leq& e^{\epsilon_i}\Pr[M(D') \in T] \\
&=& e^{\epsilon_i}\Pr[f(M(D')) \in S]
\end{eqnarray*}
which is what we wanted. To see that the same result holds for randomized mappings $f$, it suffices to observe that any randomized mapping $f:R\rightarrow R'$ is simply a convex combination of deterministic functions $f:R\rightarrow R'$.
\end{proof}

By the guarantee of differential privacy together with Fact \ref{fact:postprocessing}, we have:
\begin{eqnarray*}
\E_{x \sim f(M(D))}[u_i(x)] &=& \sum_{x \in \mathcal{A}}u_i(x)\cdot\Pr_{f(M(D))}[x] \\
&\leq& \sum_{x \in \mathcal{A}}u_i(x)\cdot\exp(\epsilon_i)\Pr_{f(M(D'))}[x] \\
&=& \exp(\epsilon_i)\E_{x \sim f(M(D'))}[u_i(x)]
\end{eqnarray*}
Similarly,
$$\E_{x \sim f(M(D))}[u_i(x)] \geq \exp(-\epsilon_i)\E_{x \sim f(M(D'))}[u_i(x)]$$

Therefore, when individual $i$ is deciding whether or not to allow his data to be used in an $\epsilon_i$-differentially private way, he is facing the decision about whether he would like his data to be used in such a way that could change his future utility by at most an additive factor of
$$\Delta u_i \equiv (\exp(\epsilon_i)-1)\E_{x \sim f(M(D'))}[u_i(x)]$$
 and so this is a natural quantity for $i$ to value his privacy at. This naturally motivates a cost function of the form $c(v_i,\epsilon_i) = (\exp(\epsilon_i)-1)v_i$, setting $v_i =\E_{x \sim f(M(D'))}[u_i(x)])$. Note that for small values of $\epsilon_i$, $(\exp(\epsilon_i)-1) \approx \epsilon_i$, which also motivates \emph{linear} utility functions of the form: $c(v_i,\epsilon_i) = \epsilon_iv_i$. Both of these types of cost functions are accommodated by our model, as well as many other reasonable choices.

\section{Characterizing Accurate Mechanisms}
\label{sec:characterize}
In this section, we show necessary and sufficient conditions on the amount of privacy that a mechanism must purchase from each individual in order to guarantee a fixed level of accuracy. To obtain a given level of accuracy, we show that a mechanism must purchase {\em at least} $\epsilon$ units of privacy, from {\em at least} $|H|$ individuals, where the values of $\epsilon$ and $|H|$ depend on the desired accuracy. We emphasize that these necessary conditions are independent of any truthfulness requirements on the mechanism, and arise purely because of the need to achieve accuracy. These conditions apply to both the sensitive value model and the insensitive value model. This greatly simplifies the mechanism-design process for auctions for private data, because it allows us to  restrict our attention to multi-unit procurement auctions without loss of generality.

\begin{theorem}
\label{thm:necessary}
Let $0 < \alpha < 1$. Any differentially private mechanism that is $\alpha\cdot n/4$-accurate must select a set of users $H \subseteq [n]$ such that:
\begin{enumerate}
\item $\epsilon_i \geq \frac{1}{\alpha n}$ for all $i \in H$.
\item $|H| \geq (1-\alpha)n$.
\end{enumerate}
\end{theorem}
\begin{proof}
Let $M$ be a mechanism that is $\alpha\cdot n/4$-accurate, and let $H \subset [n]$ be the set of individuals $i$ such that $\epsilon_i \geq 1/\alpha n$. For point of contradiction, suppose that $|H| < (1-\alpha)n$. Let $\bar{H} = [n]\setminus H$. We have that $|\bar{H}| > \alpha n$. Let $S = \{x \in \bR : |x - s| < \frac{\alpha n}{4}\}$, where $s = \sum_{i=1}^nb_i$. By the accuracy of the mechanism, we have that the estimate $\hat{s}$ output by the mechanism $M(v,D)$ satisfies:
$$\Pr[\hat{s} \in S] \geq \frac{2}{3}.$$
Let $\bar{H}^1 = \{i \in \bar{H} : b_i = 1\}$ and let $\bar{H}^0 = \{i \in \bar{H} : b_i = 0\}$. Since $\bar{H}^0$ and $\bar{H}^1$ form a partition of $\bar{H}$, it must be that $$\max(|\bar{H}^0|, |\bar{H}^1|) > \alpha n/2.$$ Without loss of generality, assume that $|\bar{H}^0| > \alpha n/2$ (the other case is identical). Let $T \subset \bar{H}^0$ such that $|T| = \alpha n/2$. Let $D'$ be the database that results in setting each bit $b_i' = b_i$ if $i \not\in T$, and $b_i' = 1$ otherwise. Note that $D'$ and $D$ have hamming distance $|T| = \alpha n/2$, and differ exactly on the indices of $T$. Let $\hat{s}'$ be the estimate generated by $M(v,D')$. By differential privacy of $M$, we have, using Fact \ref{fact:compose}:
\begin{eqnarray*}
\Pr[\hat{s}' \in S] &\geq& \exp(-\sum_{i \in T}\epsilon_i) \cdot \Pr[\hat{s} \in S] \\
&\geq& \exp(-\frac{\alpha n}{2}\cdot \frac{1}{\alpha n})\cdot \frac{2}{3} \\
&=& \frac{2}{3\sqrt{e}} \\
&>& \frac{1}{3}.
\end{eqnarray*}
Let $s' = \sum_{i = 1}^n b_i'$. Note that $s' = s + \alpha n/2$. If $\hat{s}' \in S$, then by definition: $|\hat{s}' - s| < \alpha n/4$. By the triangle inequality, we must therefore have that $|\hat{s}' - s'| > \alpha n/4$ with probability strictly greater than $1/3$, contradicting the assumption that $M$ is $\alpha\cdot n/4$ accurate.
\end{proof}
This theorem can be thought of as our main result, quantifying the necessary trade-off between accuracy and privacy: to guarantee $\alpha n/4$-accuracy, at least $(1-\alpha)$ fraction of the population {\em must} incur at least a $\frac{1}{\alpha n}$ privacy loss. The corollary below follows immediately, translating this into a lower bound on payment.
\begin{corollary}
Any $\alpha n$-accurate individually rational mechanism must pay out a total payment of at least:
$$\sum_{i=1}^np_i \geq \sum_{i=1}^{(1-4\alpha)n}c\left(v_i, \frac{1}{4\alpha n}\right)$$
where bidders are ordered such that $v_1 \leq v_2 \leq \cdots \leq v_n$.
\end{corollary}
We remark that this corollary assumes only individual rationality, and is in general achievable only by an omniscient mechanism that knows all players' cost functions. No \emph{truthful} $\alpha n$-accurate mechanism is able to pay as little as this benchmark in general.

Theorem \ref{thm:necessary} gave necessary conditions on the privacy costs of an accurate mechanism. Next, we show that up to small constant factors, they are also \emph{sufficient} conditions for an accurate mechanism:
\begin{theorem}
\label{thm:sufficient}
Let $0 < \alpha < 1$. There exists a differentially private mechanism that is $(\frac{1}{2}+\ln 3)\alpha \cdot n$-accurate and selects a set of individuals $H \subseteq [n]$ such that:
\begin{enumerate}
\item $\epsilon_i = \left\{
                      \begin{array}{ll}
                        \frac{1}{\alpha n}, & \hbox{for $i \in H$;} \\
                        0, & \hbox{for $i \not\in H$.}
                      \end{array}
                    \right.$
\item $|H| = (1-\alpha)n$.
\end{enumerate}
\end{theorem}
\begin{proof}
Let $H \subset [n]$ be any collection of individuals of size $|H| = (1-\alpha)n$, selected independently of their private bits $b_i$, and let $t = \sum_{i \in H}b_i + \alpha n/2$. Observe that for any database $D$, $|t-s| \leq \alpha n/2$. Consider the mechanism that outputs $\hat{s} = t + \textrm{Lap}(\alpha n)$. First, we claim that this mechanism is $(1/2 + \ln 3)\alpha n$-accurate. This follows by the triangle inequality conditioned on the event that Lap$(\alpha n) \leq (\ln 3)\alpha n$. It remains to verify that this holds with probability at least $2/3$. This is in fact the case:
\begin{eqnarray*}
\Pr[|\textrm{Lap}(\alpha n)| \geq (\ln 3)\alpha n]=\frac{1}{2\alpha n}\int_{-\infty}^{-(\ln 3)\alpha n} \exp\left(-\frac{|x|}{\alpha n}\right)\, dx \\
+~\frac{1}{2\alpha n}\int_{(\ln 3)\alpha n}^\infty \exp\left(-\frac{|x|}{\alpha n}\right)\, dx\\
= \frac{1}{3}.\qquad \qquad\qquad \qquad \qquad \qquad
\end{eqnarray*}
We now verify the differential privacy guarantee, which follows from the analysis given in \cite{DMNS06} of the Laplace mechanism. Let $\hat{s}$ be the estimate calculated on database $D$ (via sum $t$) and let $\hat{s}'$ be the estimate calculated on neighboring database $D^{(i)}$ (via sum $t'$). Clearly, for any $i \not\in H$ and for any $S \subset \bR$, $\Pr[\hat{s} \in S] = \Pr[\hat{s}' \in S]$ and so $\epsilon_i = 0$. Now consider some $i \in H$ and $S \subset \bR$. For any $S \subset \bR$ and $r \in \bR$, let $S - r$ denote $\{x- r : x \in S\}$.
\begin{eqnarray*}
\Pr[\hat{s}\in S] &=& \Pr[\textrm{Lap}(\alpha n) \in S - t] \\
&=&\int_{x \in S-t}\frac{1}{2\alpha n}\exp\left(-\frac{|x|}{\alpha n}\right)\, dx \\
&\leq& \exp\left(\frac{1}{\alpha n}\right) \cdot \int_{x \in S-t'}\frac{1}{2\alpha n}\exp\left(-\frac{|x|}{\alpha n}\right)\, dx \\
&=&  \exp\left(\frac{1}{\alpha n}\right) \cdot\Pr[\hat{s}'\in S]
\end{eqnarray*}
where the inequality follows from the fact that $|t-t'| \leq 1$.
\end{proof}
Theorems \ref{thm:sufficient} and \ref{thm:necessary} taken together have the effect of greatly simplifying the space of possible mechanisms for private data that we need to consider. They imply that without loss of generality (up to small constant factors in their error term), when searching for $\alpha n$-accurate mechanisms, we may restrict our attention to a special class of multi-unit procurement auctions, where we seek to purchase exactly $1/\alpha n$ units of some good (in this case, differential privacy) from exactly $(1-\alpha )n$ individuals. Once we do this, we have purchased a sufficient quantity of privacy to run the Laplace mechanism employed in Theorem \ref{thm:sufficient}, which guarantees the desired accuracy! In the next section, we consider such mechanisms.

We note that at first blush, one might expect to be able to get an accurate estimate while setting $\epsilon_i = 0$ for most individuals, by taking a random and (with high probability) representative sample of the individuals, and operating only on their sampled bits. However, note that an algorithm $A$ which takes a random sample $S$ of the population, and then runs an $\epsilon$-differentially private algorithm $A'$ on the sample $S$ results in a privacy cost $\epsilon_i > 0$ to every individual who has a non-zero probability of having been selected for the sample $S$, so long as $\epsilon > 0$. That is, sampling (and other randomization procedures) can be a part of the algorithms that we consider, but do not escape our lower bounds.

\section{Deriving Truthful Mechanisms in the Insensitive Value Model}
\label{sec:mechs}
We now give several positive results in the insensitive value model.
\subsection{Maximizing Accuracy Subject to a Budget Constraint}
\label{sec:budget}
In this section, following the characterization of accurate mechanisms in Section \ref{sec:characterize}, we restrict our attention to algorithms that guarantee $O(\alpha n)$-accuracy by purchasing $1/\alpha n$ units of privacy from exactly $(1-\alpha)n$ individuals. We consider the problem of obtaining an estimate $\hat{s}$ of maximum accuracy, subject to a hard budget constraint\footnote{This question is related to the problem of designing budget
feasible mechanisms in~\cite{Sin10,CGL10,EG11}, but differs in that our
privacy auction has {\em externalities}: a seller's cost for her good
is a function of how many other sellers are chosen as winners by the
mechanism.}: $\sum_{i=1}^n p_i \leq B$. This is a natural objective, for example, in the case of a data analyst who has $B$ dollars of grant money with which to buy data for a study, and wishes to buy the most accurate data that he can afford. We give a truthful and individually rational mechanism for this problem, and show that it is instance-by-instance optimal among the class of envy-free mechanisms.
\begin{algorithm}
\label{alg:FairQuery}
\textbf{FairQuery}$(v, D, B):$
\begin{algorithmic}
\STATE \textbf{Sort} $v$ such that $v_1 \leq v_2 \leq \ldots \leq v_n$.
\STATE \textbf{Let} $k$ be the largest integer such that $c(v_k, \frac{1}{n-k}) \leq \frac{B}{k}$.
\STATE \textbf{Output} $\hat{s} = \sum_{i=1}^kb_i + \frac{n-k}{2} + \textrm{Lap}(n-k)$
\STATE \textbf{Pay} each $i > k$ $p_i = 0$ and each $i \leq k$ $p_i = \min(\frac{B}{k}, c(v_{k+1},\frac{1}{n-k})$.
\end{algorithmic}
\end{algorithm}

We first prove that FairQuery is truthful and individually rational.
\begin{theorem}
FairQuery is truthful and individually rational, and never exceeds the data analyst's budget $B$.
\end{theorem}
\begin{proof}
First note that by the analysis from Theorem \ref{thm:sufficient}, for any $i \leq k$, $\epsilon_i = \frac{1}{n-k}$, and for any $i > k$, $\epsilon_i = 0$. For $i > k$ therefore, $p_i = c(v_i, 0) = 0$. For $i \leq k$, $p_i = \min(\frac{B}{k}, c(v_{k+1},\frac{1}{n-k})) \geq c(v_i, \frac{1}{(n-k)})$ because $c(v_i,\frac{1}{(n-k)}) \leq B/k$ by construction and $v_i \leq v_{i+1}$ by definition (recall that $c(v,\epsilon)$ is increasing in  $v$ for every $\epsilon$). Hence, individual rationality is satisfied. Note also that $\sum_{i=1}^np_i = k\cdot \min(\frac{B}{k}, c(v_{k+1},\frac{1}{n-k})) \leq B$, and so the budget constraint is also satisfied. It remains to verify truthfulness:

Fix any $v, i, v_i'$ and consider $k = k(v)$, $k' = k(v_{-i}, v_i')$, $p_i = p_i(v)$, $p'_i = p'_i(v_{-i}, v_i')$, $\epsilon_i =\epsilon_i(v)$, and $\epsilon_i' = \epsilon_i'(v_{-i}, v_i')$. There are four cases:
\begin{enumerate}
\item Case 1: $v_i' < v_i$ and $p_i > 0$. In this case, $v_i'$ moves earlier in the ordering and $\epsilon_i = \epsilon_i'$, and $p_i = p_i'$.
\item Case 2: $v_i' > v_i$ and $p_i = 0$. In this case, $v_i'$ moves later in the ordering and $\epsilon_i = \epsilon_i' = p_i = p_i' = 0$.
\item Case 3: $v_i' < v_i$ and $p_i = 0$. In this case, $v_i'$ moves earlier in the ordering, but if $p_i' > 0$ then by construction $p_i' = \min(\frac{B}{k'}, c(v_{k'+1}, \frac{1}{n-k'})) \leq c(v_i, \frac{1}{(n-k')})$. This follows because $k'$ is such that $v_{k'+1} \leq v_i$ for all $i > k$ such that $p_i' > 0$.
\item Case 4: $v_i' > v_i$ and $p_i > 0$. In this case, $v_i'$ moves later in the ordering, and either $p_i' = p_i$ and $\epsilon_i' = \epsilon_i$, or $p_i' = 0$ and $\epsilon_i = 0$.  In the second case, by individual rationality, $p_i - c(v_i, \epsilon_i) \geq 0 = p_i' - c(v_i,\epsilon_i')$.
\end{enumerate}
Thus in all four cases, deviations are not beneficial, and the mechanism is truthful.
\end{proof}

The next natural question to ask is: does FairQuery guarantee the data analyst a good level of accuracy, given his budget? As is always the case in prior-free mechanism design, it is important to specify what our benchmark is -- good compared to what? Because mechanisms of the kind that we are considering always buy the same amount of privacy from an individual from whom they buy any privacy at all, a natural benchmark to consider is the set of all ``envy-free'' mechanisms which guarantee that no individual would prefer the outcome granted to any other.

\begin{definition}
A mechanism for private data is \emph{envy-free} if for all possible valuation vectors $v$, and for all individuals $i, j$, $p_i-c(v_i,\epsilon_i) \geq p_j - c(v_i, \epsilon_j)$. That is, after the mechanism has determined the privacy costs and payments to each individual, there are no individuals who would prefer to have the payment and privacy cost granted to any other individual.
\end{definition}
\begin{observation}
Any truthful envy-free mechanism which buys either no privacy or $\epsilon$-privacy from each individual (\ie, if $\epsilon_i > 0, \epsilon_j > 0$ then $\epsilon_i = \epsilon_j$) must have the property that for all $i,j$ with $\epsilon_i > \epsilon_j > 0$, $p_i = p_j$. Call such mechanisms \emph{fixed purchase mechanisms}. That is, envy free fixed purchase mechanisms must pay each individual from whom privacy is purchased the same fixed price.
\end{observation}
Note that by the characterization in Section \ref{sec:characterize}, we may restrict ourselves to considering fixed purchase mechanisms essentially without loss of generality (we may lose only a small constant factor in the approximation factor).
Therefore we can compare our mechanism to the envy free benchmark:
\begin{proposition}
For any set of valuations $v\in \bR_+^n$ (\ie, on an instance-by-instance basis) FairQuery achieves the optimal accuracy given budget $B$, among the set of all truthful, individually rational envy-free fixed purchase mechanisms.
\end{proposition}
\begin{proof}
First, observe the easy fact that FairQuery is indeed an envy free fixed purchase mechanism. We then merely observe that for any vector of valuations $v$, if FairQuery sets $\epsilon_i > 0$ for $k$ individuals, then by the definition of $k$, it must be that $c(v_{k+1}, \frac{1}{(n-k-1)}) > \frac{B}{k+1}$, and so any mechanism that set $\epsilon_i > 0$ for $k'$ individuals for $k' > k$ must have $p_{k+1} > \frac{B}{(k+1)}$ by individual rationality. But by envy-freeness, it must have $p_i = p_{k+1} > \frac{B}{(k+1)}$ for all $i \leq k$. But in this case, we would have
$$\sum_{i=1}^np_i \geq k'\cdot p_{k+1} > (k+1)\cdot \frac{B}{k+1} = B$$
which would violate the budget constraint.
\end{proof}

\subsection{Minimizing Payment Subject to an Accuracy Constraint}
In this section, we consider mechanisms for the dual goal of truthfully obtaining a $k$-accurate estimate for some fixed accuracy constraint $k$ while minimizing the payment required. Again, we restrict ourselves to the model of multi-unit procurement auctions justified in Section \ref{sec:characterize}. In this setting, we show that the VCG mechanism is in fact optimal.

Recall that for a fixed accuracy goal $\alpha n$, by Theorem \ref{thm:sufficient}, it is sufficient to buy $\frac{(1/2 + \ln 3)}{\alpha n}$ units of privacy from $(1-\frac{\alpha}{(1/2 + \ln 3)})n$ people. We may therefore view our setting as a multi-unit procurement auction in which every individual is selling a single good ($\frac{(1/2 + \ln 3)}{\alpha n}$ units of privacy), for which they have valuation $c_i(v_i, \frac{(1/2 + \ln 3)}{\alpha n})$. The constraint on accuracy simply states that we must buy $(1-\frac{\alpha}{(1/2 + \ln 3)})n$ units of the good. In this case, we can analyze a simple application of the standard VCG mechanism:

\begin{algorithm}
\label{alg:VCG}
\textbf{MinCostAuction}($v, D, \alpha$):
\begin{algorithmic}
\STATE \textbf{Let} $\alpha' = \frac{\alpha}{1/2 + \ln 3}$ and $k = \lceil(1-\alpha')n\rceil$.
\STATE \textbf{Let} $w_i = c(v_i, \frac{1}{n-k})$.
\STATE \textbf{Sort} $w_i$ such that $w_1 \leq w_2 \leq \ldots \leq w_n$.
\STATE \textbf{Output} $\hat{s} = \sum_{i=1}^kb_i + \frac{n-k}{2} + \textrm{Lap}(\alpha'n)$
\STATE \textbf{Pay} each $i > k$ $p_i = 0$ and each $i \leq k$ $p_i = w_{k+1}$.
\end{algorithmic}
\end{algorithm}
We first show that MinCostAuction does indeed satisfy the constraints of truthfulness and individual rationality, while obtaining sufficient accuracy.
 \begin{proposition}
 MinCostAuction is truthful, individually rational and $\alpha n$-accurate.
 \end{proposition}
 \begin{proof}
 That MinCostAuction is $\alpha n$-accurate follows immediately from Theorem \ref{thm:sufficient}. Moreover, by Theorem \ref{thm:sufficient}, for each $i \leq k$, $\epsilon_i = 1/(\alpha' n)$ and for $i > k$, $\epsilon_i =0$. Truthfulness and individual rationality then follow immediately from the fact that each $w_i = c_i(v_i, 1/(\alpha'n))$, that $c_i(\cdot, \epsilon)$ is an increasing function of $v_i$, and that MinCostAuction is an instantiation of the classical VCG mechanism.
\end{proof}
MinCostAuction achieves its target utility at a cost of $\sum_{i=1}^np_i = k\cdot w_{k+1}$. We now show that no other envy-free multi-unit procurement auction with the same accuracy guarantees (i.e. one that guarantees buying $k$ units) makes smaller payments than MinCostAuction.
\begin{theorem}
No truthful, individually rational, envy-free fixed purchase auction that guarantees purchasing $k$ units can have total payment less than $k\cdot w_{k+1}$.
\end{theorem}
\begin{proof}
For the sake of contradiction, suppose we have such a mechanism $M$. Fix some vector of valuations $v$ that yields payments $p(v)$ such that $\sum_{i=1}^np_i(v) < k\cdot w_{k+1}$ (Recall that $w_i = c(v_i, \frac{1}{n-k})$). First, if it is not already the case, we will construct a bid profile such that an item is purchased from some seller who is not among the $k$ lowest bidders. It must be that there exists some $i$ such that an item is purchased from $i$ at a price of $p_i$, such that $w_i \leq p_i < w_{k+1}$ (otherwise $\sum_{i=1}^np_i(v) \geq k\cdot w_{k+1}$). Let $v' = (v_{-i}, v_i')$ where $c(v_i', \frac{1}{n-k}) = (p_i + w_{k+1})/2$, be a bid profile in which bidder $i$ raises his reported value to be above $p_i$ while remaining below $w_{k+1}$. Note that such a $v_i'$ exists since $c$ is assumed to be continuous. Let $p' = p'(v)$ be the new payment vector. By individual rationality and truthfulness, it must be that in this new bid profile $v'$, player $i$ is no longer allocated an item: by individual rationality, he would have to be paid $p'_i > p_i$ if he were allocated an item, but if his true valuation were $w_i$, then this would be a beneficial deviation, contradicting truthfulness. Because the mechanism is constrained to always buy at least $k$ items, it must be that in $v'$, an item is now purchased from some seller $j$ such that $j \geq k+1$. By individual rationality, $p'_j \geq w_{j} \geq w_{k+1}$. But by envy-freeness, it must be that for every seller $i$ from whom an item was purchased, $p'_i = p'_j \geq w_{k+1}$. Because at least $k$ items are purchased, we therefore have $\sum_{i=1}^np'_i \geq k\cdot w_{k+1}$, which contradicts the purported payment guarantee of mechanism $M$.
\end{proof}

\section{Truthful Mechanisms in the Sensitive Value Model}
\label{sec:privatebids}
In Section \ref{sec:mechs}, we considered truthful, individually rational mechanisms that compensated users for the privacy loss due to the mechanisms' use of the individual's private bits $b_i$, but \emph{not} due to the mechanisms' use of their valuations for privacy, $v_i$. Nevertheless, as we observed in the introduction, it is quite reasonable to assume that individual's valuations for privacy are correlated with their private bits. Can we design mechanisms that treat individuals' valuations for privacy as private data as well, and compensate individuals for the privacy loss due to the use of their valuations $v_i$? In this section, we show that the answer is generically `no' if we allow individuals to have arbitrarily high valuations for privacy. Moreover, we note that if we try to impose an a-priori bound on individual's valuations for privacy, then we re-introduce the same source of sampling bias that we had hoped to solve by running an auction.

Recall that a mechanism has two outputs: the estimate $\hat{s}$, and the payment $P$ that the data analyst must make. Note that if the bids are private data as well (i.e. if we are in the sensitive value model), then a mechanism which is $\epsilon_i$-differentially private with respect to bidder $i$ must satisfy, for every set of estimate/payment tuples $S \subset \bR_+^2$  and for each $(v, D) \in \bR_+^n\times\{0,1\}^n$, $\Pr[M(v,D) \in S] \leq \exp(\epsilon_i)\Pr[M(v^{(i)},D^{(i)}) \in S]$, where $v^{(i)}$ and $D^{(i)}$ are arbitrary vectors that are identical to $v$ and $D$ everywhere except possibly on their $i$th index.

\begin{theorem}
\label{thm:impossible}
If bidder valuations for privacy may be arbitrarily large (\ie, $v \in \bR_+^n$) then no individually rational direct revelation mechanism $M$ can protect the privacy of the bidder valuations and promise $k$-accuracy for any $k < n/2$ (\ie, any nontrivial value).
\end{theorem}
\begin{proof}
For simplicity, consider bidders with linear cost functions: $c_i(v_i,\epsilon_i) = v_i\cdot \epsilon_i$. Assume that $M$ is $k$-accurate for some $k < n/2$. Run the mechanism $M(v, D)$ and obtain an estimate $\hat{s}$ and privacy costs $\epsilon_i$ for each $i \in [n]$. Let $P = \sum_{i=1}^np_i$ be the payment that the data analyst makes. By individual rationality, $P \geq \sum_{i=1}^nv_i\epsilon_i \geq \min_{i}v_i\cdot \sum_{i=1}^n\epsilon_i$. We trivially have that either $\Pr[\hat{s} \in [0,n/2)] \geq 1/2$ or $\Pr[\hat{s} \in [n/2, n]] \geq 1/2$. Without loss of generality, assume $\Pr[\hat{s} \in [0,n/2)] \geq 1/2$. Let $D' = 1^n$, and let $\hat{s}'$ be the estimate obtained by running $M(v,D')$. By accuracy, we have that: $\Pr[\hat{s}' \in (n/2, n]] \geq \frac{2}{3}$. However, by differential privacy, together with Fact \ref{fact:compose} we have:
{\footnotesize
$$\frac{2}{3} \leq \Pr[\hat{s}' \in (n/2, n]] \leq \exp(\sum_{i=1}^n\epsilon_i)\Pr[\hat{s} \in (n/2, n]]$$$$ \leq \frac{\exp(\sum_{i=1}^n\epsilon_i)}{2}$$
}
Solving, we find that $\sum_{i=1}^n\epsilon_i \geq \ln(4/3)$, independent of $v$. We therefore have by individual rationality that $\Pr[P \in [0,\ln(4/3)\min_iv_i)] = 0$. By differential privacy, this must hold simultaneously for all inputs to the mechanism $(v,D)$: that is, such a mechanism can not charge a finite price $P$ for \emph{any} input, which completes the proof.
\end{proof}
\begin{remark}
A natural (partial) way around the impossibility result of Theorem \ref{thm:impossible} is to restrict bidder valuations to lie in a bounded range (e.g. $[0,1]$). This is unsatisfying, however, because it re-introduces the very source of sampling bias that we wanted to solve by running an auction. That is, bidders who happen to value their privacy at a higher rate than allowed by the mechanism will simply not participate in the auction, which might systematically skew the resulting estimate in a way that we cannot measure.
\end{remark}

\section{Conclusion and Future Directions}
\label{sec:future}
The main contribution of this paper is to formalize the notion of auctions \emph{for} private data, and to show that the design space of such auctions can without loss of generality be taken to be the simple setting of multi-unit procurement auctions. This initiates an intriguing new area of study that raises many questions. Among these are:
\begin{enumerate}
\item What is the proper benchmark for auctions in our setting? In this paper, we used the class of fixed-price (or envy free) mechanisms, which has become standard in the field of prior-free mechanism design \cite{HR08, HK07}. Is there a more natural benchmark?
\item We have shown that generically, no direct revelation mechanism can compensate individuals for the loss of privacy which results from correlations between their private data and their reported costs for privacy. Nevertheless, such correlations exist! It is unsatisfying to restrict individual valuations for privacy to lie in a bounded range, because this reintroduces the very source of bias that we hoped to overcome by designing auctions. However, is there some restricted sense in which we can protect (and compensate users for) the privacy of their \emph{valuations for privacy}? This requires the development of new models.
\item We have assumed throughout this paper that the private bits of the users, $b_i$ are already known to some database administrator such as a hospital, or are otherwise verifiable. Although this is a natural assumption in some settings, what if it does not hold? Is there any way to mediate the purchase of private data directly from individuals who have the power to lie about their private data?
\item In this paper we considered an extremely simple market, in which there was a single data analyst wanting to buy data from a population. How about a two sided market, in which there are multiple data analysts, competing for access to the private data from multiple populations? Can we privately compute the market clearing prices for access to data in this way?
\item In this paper we considered a one-shot mechanism. In reality, the administrator of a private database will face multiple requests for access to his data as time goes on. How should the data analyst reason about these online requests and his value for the marginal privacy loss that he will incur after answering each request?
\end{enumerate}

\section*{Acknowledgements}
This paper has benefitted from discussions with many people. In particular, we would like to thank Frank McSherry for suggesting the problem of treating differential privacy as a currency to be bought and sold at the 2010 IPAM workshop on differential privacy, Tim Roughgarden for helpful early discussions on benchmarks for prior-free mechanisms for private data, and Sham Kakade, Ian Kash, Michael Kearns, Katrina Ligett, Mallesh Pai, Ariel Procaccia, David Parkes, Kunal Talwar, Salil Vadhan, Jon Ullman, and Bumin Yenmez for helpful comments and discussions. We would also like to thank Yu-Han Lyu for pointing out several typos in an earlier version of the paper.

\bibliographystyle{abbrv}
\newcommand{\etalchar}[1]{$^{#1}$}

\end{document}